\documentclass[runningheads]{llncs}
\usepackage{graphicx}
\usepackage{subfig}

\usepackage{amssymb}
\usepackage{amsmath}
\usepackage{amsfonts}
\usepackage{latexsym}
\usepackage{verbatim}
\usepackage{enumerate}
\usepackage{url}
\usepackage[linesnumbered,noend]{algorithm2e}
\usepackage{color}
\usepackage{hyphenat}
\newcommand{\Ch}{\Children}
\newcommand{\Children}{\ensuremath{{\mathrm{Ch}}}}

\newcommand{\parent}{\ensuremath{{\rm parent}}}
\newcommand{\pred}{\ensuremath{{\rm pred}}}

\newcommand{\QED}{\hfill$\Box$}

\newcommand{\DG}{\ensuremath{H_\P}} 

\newcommand{\Cl}{\ensuremath{{\mathrm{Cl}}}} 		

\newcommand{\NULL}{\ensuremath{\mathtt{null}}}
\newcommand{\MP}{\ensuremath{M_\P}} 
\newcommand{\pos}{\ensuremath{\pi}}
\newcommand{\initPos}{\ensuremath{\pi^\mathrm{init}}}
\newcommand{\posSet}{\ensuremath{\Pi}}

\newcommand{\enqueue}{\ensuremath{\textsc{enqueue}}}
\newcommand{\dequeue}{\ensuremath{\textsc{dequeue}}}
\newcommand{\Build}{\ensuremath{\textsc{Build}}} 

\newcommand{\Oh}{\ensuremath{\mathcal{O}}}
\renewcommand{\P}{\ensuremath{\mathcal{P}}}

\newcommand{\T}{\ensuremath{\mathcal{T}}}

\newcommand{\Q}{\ensuremath{\mathcal{Q}}}
\newcommand{\CS}{\ensuremath{A}}

\newcommand{\expose}{\ensuremath{\mathtt{appear}}} %
\newcommand{\DISAGREEMENT}{\ensuremath{\mathtt{disagreement}}}

\newcommand{\LCA}{\ensuremath{\mathrm{LCA}}} 

\SetKw{KwOr}{or}%
\SetKw{KwAnd}{and}%
\SetKw{KwContinue}{continue}%

\SetKwProg{Fn}{}{}{}
\SetKwFunction{Build}{Build}%
\SetKwFunction{BuildAST}{BuildAST}%
\SetKwFunction{MergeComponents}{MergeComponents}
\SetKwFunction{HandleSingleComp}{HandleSingleComponents}
\SetKwFunction{ComputeSuccessor}{GetDecomposition}
\SetKwFunction{BuildAST}{BuildAST}

\newtheorem{assumption}{Assumption}
\newtheorem{observation}{Observation}

\begin{document}
%
\title{Testing the Agreement of Trees with Internal Labels
}
%
%
\author{David Fern\'{a}ndez-Baca
\and
Lei Liu} 
%
\authorrunning{D.\ Fern\'{a}ndez-Baca, L.\ Liu}
%

\institute{Department of Computer Science, Iowa State University, Ames IA 50011, USA
\email{\{fernande,lliu\}@iastate.edu}} 
%
\maketitle              
%



\begin{abstract}
The input to the agreement problem is a collection $\P = \{\T_1, \T_2, \dots , \T_k\}$ of phylogenetic trees, called input trees, over partially overlapping sets of taxa.  The question is whether there exists a tree $\T$, called an agreement tree, whose taxon set is the union of the taxon sets of the input trees, such that for each $i \in \{1, 2, \dots , k\}$, the restriction of $\T$ to the taxon set of $\T_i$ is isomorphic to $\T_i$.  We give a $\Oh(n k (\sum_{i \in [k]} d_i + \log^2(nk)))$ algorithm for a generalization of the agreement problem in which the input trees may have internal labels, where $n$ is the total number of distinct taxa in \P, $k$ is the number of trees in \P, and $d_i$ is the maximum number of children of a node in $\T_i$.

\keywords{Phylogenetic tree \and Taxonomy  \and Agreement \and Algorithm.}
\end{abstract}

%
%
\section{Introduction}

In the \emph{tree agreement problem} (\emph{agreement problem}, for short), we are given a collection $\P = \{\T_1, \T_2,  \ldots, \T_k\}$ of rooted phylogenetic trees with partially overlapping taxon sets. $\P$ is called a \emph{profile} and the trees in $\P$ are the \emph{input trees}.  The question is whether there exists a tree $\T$ whose taxon set is the union of the taxon sets of the input trees, such that, for each $i \in \{1, 2, \dots , k\}$, $\T_i$ is isomorphic to the restriction of $\T$ to the taxon set of $\T_i$.  If such a tree $\T$ exists, then we call $\T$ an \emph{agreement tree} for $\P$ and say that $\P$ \emph{agrees}; otherwise, $\P$ \emph{disagrees}.   
 The first explicit polynomial-time algorithm for the agreement problem is in reference  \cite{NgWormald96}\footnote{These authors refer to what we term ``agreement'' as ``compatibility''.  What we call ``compatibility'', they call ``weak compatibility''.}.  The agreement problem can be solved in $O(n^2 k)$ time, where $n$ is the number of distinct taxa in $\P$ \cite{FBGSV2015}.   

Here we study a generalization of the agreement problem, where the internal nodes of the input trees may also be labeled. These labels represent higher-order taxa; i.e., in effect, sets of taxa.  Thus, for example,   
an input tree may contain the taxon \emph{Glycine max} (soybean) nested within a subtree whose root is labeled Fabaceae (the legumes), itself nested within an Angiosperm subtree.  Note that leaves themselves may be labeled by higher-order taxa.  We present a $\Oh(n k (\sum_{i \in [k]} d_i + \log^2(nk)))$ algorithm for the agreement problem for trees with internal labels, where $n$ is the total number of distinct taxa in \P, $k$ is the number of trees in \P, and, for each $i \in \{1, 2, \dots , k\}$, $d_i$ is the maximum number of children of a node in $\T_i$.


\paragraph{Background.}  
A close relative of the agreement problem is the \emph{compatibility problem}.  The input to the compatibility problem is a profile $\P = \{\T_1, \T_2,  \ldots, \T_k\}$ of rooted phylogenetic trees with partially overlapping taxon sets.  The question is whether there exists a tree $\T$ whose taxon set is the union of the taxon sets of the input trees such that each input tree $\T_i$ can be obtained from the restriction of $\T$ to the taxon set of $\T_i$ through edge contractions.
If such a tree $\T$ exists, we refer to $\T$ as a \emph{compatible tree} for $\P$ and say that $\P$ is \emph{compatible}; otherwise, $\P$ is \emph{incompatible}. 
Compatibility is a less stringent requirement than agreement; therefore, any profile that agrees is compatible, but the converse is not true. 
The compatibility problem for phylogenies (i.e., trees without internal labels), is solvable in $\Oh(\MP \log^2 \MP)$ time, where $\MP$ is the total number of nodes and edges in the trees of $\P$ \cite{DengFB2017}. Note that $\MP = \Oh(nk)$.

Compatibility and agreement reflect two distinct approaches to dealing with \emph{multifurcations}; i.e., non-binary nodes, also known as  \emph{polytomies}.  Suppose that node $v$ is a multifurcation in some input tree of $\P$ and that $\ell_1$, $\ell_2$, and $\ell_3$ are taxa in three distinct subtrees of $v$.  In an agreement tree for $\P$, these three taxa must be in distinct subtrees  of some node in the agreement tree.  In contrast, a compatible tree for $\P$ may contain no such node, since a compatible tree is allowed to ``refine'' the multifurcation at $v$ --- that is, group two out of $\ell_1$, $\ell_2$, and $\ell_3$ separately from the third.  
Thus, compatibility treats multifurcations  as ``soft'' facts;  agreement treats them as ``hard'' facts
\cite{Maddison89}.   Both viewpoints can be valid, depending on the circumstances. 

The agreement and compatibility problems are fundamental special cases of the \emph{supertree problem}, the problem of synthesizing a collection of phylogenetic trees with partially overlapping taxon sets into a single supertree that represents the information in the input trees \cite{BinindaEmonds04,Baum:1992,Ragan:1992,WarnowSupertrees2018}.   
The original supertree methods were limited to input trees where only the leaves are labeled, but there has been increasing interest in incorporating internally labeled trees in supertree analysis, motivated by the desire to incorporate \emph{taxonomies} in these analyses. Taxonomies group organisms according to a system of taxonomic rank  (e.g., family, genus, and species); two examples are the NCBI taxonomy \cite{NCBI2009} and the Angiosperm taxonomy \cite{APG2016}.  Taxonomies 
provide structure and completeness that can be hard to obtain otherwise \cite{Page2004,HinchliffPNAS2015,RedelingsHolder2017}, offering a way to circumvent one of the obstacles to building comprehensive phylogenies: the limited taxonomic overlap among different phylogenetic studies \cite{Sanderson:2008}. 

Although internally labeled trees, and taxonomies in particular, are not, strictly speaking, phylogenies, they have many of the same mathematical properties as phylogenies.  Both phylogenies and  internally labeled trees are \emph{$X$-trees} (also called \emph{semi-labeled trees}) 
\cite{BordewichEvansSemple2006,SempleSteel03}.  Algorithmic results for compatibility and agreement of internally labeled trees are scarce, compared to what is available for ordinary phylogenies.  
To our knowledge, the first algorithm for testing compatibility of internally labeled trees is in  \cite{DanielSemple2004} (see also \cite{BerrySemple2006}).   The fastest known algorithm for the problem runs in $\Oh(\MP \log^2 \MP)$ time \cite{DengFB2017b}.  We are unaware of any previous algorithmic results for the agreement problem for internally labeled trees.

All algorithms for compatibility and agreement that we know of are indebted  to Aho et al.'s \Build  algorithm \cite{AhoSagivSzymanskiUllman81}.  The time bounds for agreement algorithms are higher than those of compatibility algorithms, due to the need for agreement trees to respect the multifurcations in the input trees.   To handle agreement, \Build has to be modified so that certain sets of the partition of the taxa it generates are re-merged to reflect the multifurcations in the input trees, adding considerable overhead \cite{NgWormald96,FBGSV2015} (similar issues are faced when testing consistency of triples and fans \cite{JanssonRECOMB2017}). This issue becomes more complex for internally labeled trees, in part because internal nodes with the same label, but in different trees, may jointly imply multifurcations, even if all input trees are binary.



\paragraph{Organization of the paper.}
Section \ref{sec:prelims} provides a formal definition of the agreement problem for internally labeled trees.  Section \ref{sec:positions} studies the decomposability properties of profiles that agree. These properties allow us to reduce an agreement problem on a profile into independent agreement problems on subprofiles, leading to the agreement algorithm presented in Section \ref{sec:constructingATs}.  Section \ref{sec:discussion} contains some final remarks.  All proofs are in the Appendix.


%
%

\section{Preliminaries}\label{sec:prelims}

For each positive integer $r$, $[r]$ denotes the set $\{1, \dots , r\}$.

\paragraph{Graphs and trees.} 
Let $G$ be a graph. $V(G)$ and $E(G)$ denote the node and edge sets of $G$. Let $U$ be a subset of $V(G)$. Then the \emph{subgraph of $G$ induced by $U$} is the graph whose vertex set is $U$ and whose edge set consists of all of the edges in $E(G)$ that have both endpoints in $U$.

A \emph{tree} is an acyclic connected graph.
All trees here are assumed to be rooted.  For a tree $T$, $r(T)$ denotes the root of $T$. 
Suppose $u, v \in V(T)$.  Then, $u$ is an \emph{ancestor} of $v$ in $T$, denoted $u \le_T v$, if $u$ lies on the path from $v$ to $r(T)$ in $T$.  If $u \le_T v$, then $v$ is a \emph{descendant} of $u$.  Node %
$u$ is a \emph{proper ancestor} of $v$, denoted $u <_T v$, if $u \le_T v$ and $u\neq v$. If $\{u,v\} \in E(T)$ and $u \le_T v$, then  $u$ is the \emph{parent} of $v$ and $v$ is a \emph{child} of $u$. 
For each $x \in V(T)$, we use $\parent_T(x)$, and $\Ch_T(x)$, $T(x)$ to denote the parent of $x$, the children of $x$, and the subtree of $T$ rooted at $x$, respectively.  We extend the child notation to subsets of $V(T)$ in the natural way: for $U \subseteq V(T)$, $\Ch_T(U) = \bigcup_{u \in U} \Ch_T(u)$.  Thus, if $U = \emptyset$, then $\Ch_T(U) = \emptyset$.

Let $T$ be a tree and suppose $U \subseteq V(T)$.
The \emph{lowest common ancestor of $U$ in $T$}, denoted $\LCA_T(U)$, is the unique smallest upper bound of $U$ under $\le_T$.


\paragraph{$X$-trees\label{sec:XTrees}.} 

Throughout the paper, $X$ denotes a set of \emph{labels} (that is, taxa, which may be, e.g., species or families of species).
An \emph{$X$-tree} is a pair $\T = (T,\phi)$ where $T$ is a tree and $\phi$ is a mapping from $X$ to $V(T)$ such that, for every node $v \in V(T)$ of degree at most two, $v \in \phi(X)$.   $X$ is the \emph{label set} of $\T$ and $\phi$ is the \emph{labeling function} of $\T$.  
For every node $v \in V(T)$, $\phi^{-1}(v)$ denotes the (possibly empty) subset of $X$ whose elements map into $v$; these elements as the \emph{labels of $v$}. If $\phi^{-1}(v) \neq \emptyset$, then $v$ is \emph{labeled}; otherwise, $v$ is \emph{unlabeled}.  

By definition, every leaf in an $X$-tree is labeled, and any node, including the root, that has a single child must be labeled.  Nodes with two or more children may be labeled or unlabeled.  An $X$-tree $\T = (T,\phi)$ is \emph{singly labeled} if every node in $T$ has
at most one label; $\T$ is \emph{fully labeled} if every node in $T$ is labeled.  

$X$-trees, also known as \emph{semi-labeled trees}, generalize ordinary phylogenetic trees (also known as \emph{phylogenetic $X$-trees} \cite{SempleSteel03}).  An ordinary phylogenetic tree is a semi-labeled tree $\T = (T,\phi)$ where $r(T)$ has degree at least two and $\phi$ is a bijection from $L(\T)$ into leaf set of $T$ (thus, internal nodes are not labeled). 



Let  $\T = (T,\phi)$ be an $X$-tree. For each $u \in V(T)$,  $X(u)$ denotes the set of all labels in the subtree of $T$ rooted at $u$; that is, $X(u) = \bigcup_{v: u \le_T v} \phi^{-1}(v)$. $X(u)$ is called a \emph{cluster} of $T$.  
$\Cl(\T)$ denotes the set of all clusters of $\T$.   
We extend the cluster notation to sets of nodes as follows.  Let $U$ be a subset of $V(T)$.  Then, $X(U) = \bigcup_{v \in U} X(v)$.  If $U = \emptyset$, then $X(U) = \emptyset$.
 

Suppose $Y \subseteq X$ for an $X$-tree $\T = (T,\phi)$.  The \emph{restriction} of $\T$ to $Y$, denoted $\T |Y$, is the semi-labeled tree whose cluster set is
$\Cl(\T | Y) = \{W \cap Y : W \in \Cl(\T) \text{ and } W \cap Y \neq \emptyset \}.$
Intuitively, $\T | Y$ is obtained from the minimal rooted subtree of $T$ that connects the nodes in $\phi(Y)$ by suppressing all vertices $v$ such that $v \notin \phi(Y)$ and $v$ has only one child.

Let $\T = (T,\phi)$ be an $X$-tree and $\T' = (T', \phi')$ be an $X'$-tree such that $X' \subseteq X$.  $\T$  \emph{agrees with} $\T'$ if $\Cl(\T') = \Cl(\T|X')$.  It is well known that the clusters of a tree determine the tree, up to isomorphism \cite[Theorem 3.5.2]{SempleSteel03}.  Thus, $\T$  agrees with $\T'$ if $\T'$ and $\T|X'$ are isomorphic.

\paragraph{Profiles and agreement.}

Throughout the rest of this paper, $\P$ denotes a set $\{\T_1, \T_2, \dots, \T_k\}$ such that, for each $i \in [k]$, $\T_i = (T_i, \phi_i)$ is a phylogenetic $X_i$-tree for some set $X_i$ (Figure \ref{fig:profile}). We refer to \P\ as a \emph{profile}, and to the trees in \P\ as \emph{input trees}.  
We write $X_\P$ to denote $\bigcup_{i \in [k]} X_i$.

\begin{figure}[t]
\centering
\subfloat[\label{fig:profile}]{%
  \includegraphics[width=0.65\textwidth]{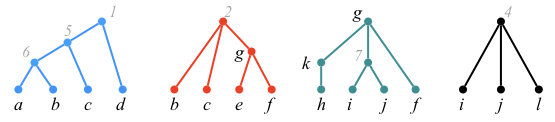}%
}\hfil
\subfloat[\label{fig:tree}]{%
  \includegraphics[width=0.24\textwidth]{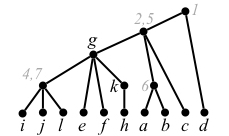}%
}

\caption{
(a) A profile $\P= \{\T_1, \T_2, \T_3, \T_4\}$. 
(b) An agreement tree for $\P$.} 
\end{figure}


A profile $\P$ \emph{agrees} if there is an $X_\P$-tree $\T$ that agrees with each of the trees in $\P$.  If $\T$ exists, we refer to $\T$ as an \emph{agreement tree for $\P$}. 
See Figure \ref{fig:tree}. 



Given a subset $Y$ of $X_\P$, the \emph{restriction} of $\P$ to $Y$, denoted $\P|Y$, is the profile defined as 
$\P|Y = \{\T_1|Y \cap X_1, \T_2|Y \cap X_2, \dots, \T_k|Y \cap X_k\}.$
The proof of the following lemma is straightforward.

\begin{lemma}\label{lm:compatSubprofile}
Suppose a profile $\P$ has an agreement tree $\T$.  Then, for any $Y \subseteq X_\P$, $\T | Y$ is an agreement tree for $\P| Y$.
\end{lemma}

Suppose $\P$ contains trees that are not fully labeled.  We can convert $\P$ into an equivalent profile $\P'$ of fully-labeled trees as follows.  For each $i \in [k]$, let $l_i$ be the number of unlabeled nodes in $T_i$.  Create a set $X'$ of $n' = \sum_{i \in [k]} l_i$ labels such that $X' \cap X_\P = \emptyset$.  For each $i \in [k]$ and each $v \in V(T_i)$ such that $\phi_i^{-1}(v) = \emptyset$, make $\phi_i^{-1}(v) = \{\ell\}$, where $\ell$ is a distinct element from $X'$.
We refer to $\P'$ as the \emph{profile obtained by adding distinct new labels to $\P$} (see Figure \ref{fig:profile}).

\begin{lemma}
\label{lm:fullyL}
Let $\P'$ be the profile obtained by adding distinct new labels to $\P$. Then, $\P$ agrees if and only if $\P'$ agrees. Further, if $\T$ is an agreement tree for $\P'$, then $\T$ is also and agreement tree for $\P$.
\end{lemma}

From this point forward, we make the following assumption. 

\begin{assumption}\label{as:fullySingular}
For each $i \in [k]$, $\T_i$ is fully and singularly labeled.
\end{assumption}

By Lemma \ref{lm:fullyL}, no generality is lost in assuming that all trees in $\P$ are fully labeled. 
The assumption that the trees are singularly labeled is inessential; it is only for clarity.  Note that, even with the latter assumption, a tree that agrees with $\P$ is not necessarily singularly labeled.  Figure \ref{fig:tree} illustrates this fact.

\begin{lemma}\label{lemma:profileEmbedding}
If profile $\P$ agrees, then $\P$ has an agreement tree $\T = (T,\phi)$ such that $\phi^{-1}(v) \neq \emptyset$ for each node $v \in V(T)$. 
\end{lemma}
%
%

By Assumption \ref{as:fullySingular}, for each $i \in [k]$, there is a bijection between the labels in $X_i$ and the nodes of $V(T_i)$.  For this reason, we will often refer to nodes by their labels.  In particular, given a label $\ell \in X_i$, we write $X_i(\ell)$ to denote $X_i(\phi_i(\ell))$ (the cluster of $\T_i$ at the node labeled $\ell$), $\Ch_{T_i}(\ell)$ to denote $\phi_i(\Ch_{T_i}(\phi_i(\ell))$ (the labels of children of $\ell$ in $\T_i$), and, for $A \subseteq X_i$, $\Ch_{T_i}(A)$ to denote $\phi_i(\Ch_{T_i}(\phi_i(A))$. 


The following characterization of agreement generalizes a result in \cite{FBGSV2015}.

\begin{lemma}\label{lemma:embedding}
Let $\P$ be a profile and $\mathcal{T} = (T, \phi)$ be an $X_\P$-tree. Then, $\T$ is an agreement tree for $\P$ if and only if, for each $i \in [k]$, there exists a function $\phi_i : X_i \rightarrow V(T)$ such that for every label $a \in X_i$,
\begin{enumerate}[(E1)]
\item
$\phi_i(a) = \LCA_T(X_i(a))$,
\item
for each label $b \in \Ch_{T_i}(a)$, $\phi_i(a) <_T \phi_i(b)$, and 
\item
for every two distinct labels $b, c  \in \Ch_{T_i}(a)$, there exist distinct nodes $u, v \in \Ch_T(\phi_i(a))$ such that $\phi_i(b) \in X_\P(u)$ and $\phi_i(c) \in X_\P(v)$.
\end{enumerate}
\end{lemma}

We refer to a function $\phi_i$ satisfying conditions (E1)--(E3) of Lemma \ref{lemma:embedding} as a \emph{topological embedding} of $\T_i$ into $\T$. Observe that, by transitivity, condition (E2) implies that, for any $a, b \in X_i$, if $a <_{T_i} b$, then $\phi_i(a) <_T \phi_i(b)$.

\section{Positions in a Profile\label{sec:positions}}

A \emph{position} in a profile $\P$ is a tuple $\pos = (\pos_1, \pos_2, \dots, \pos_k)$ where, for each $i \in [k]$, either $\pos_i = \emptyset$ or $\pos_i = \{\ell\}$, for some $\ell \in X_i$. 
Note that the definition of a position allows for the possibility that there exist $i, j \in [k]$, $i \neq j$, such that $\ell \in \pos_i$, but $\ell \notin \pos_{j}$, even if $\ell \in X_i$ and $\ell \in X_{j}$. 
At any given point during its execution, our agreement algorithm focuses on testing the agreement of the subprofile of \P\ determined by the subtrees associated with a specific position.  

For a position $\pos$ in \P, let $X_\P(\pos)$ denote the set of labels $\bigcup_{i \in [k]} X_i(\pos_i)$.   A label $\ell \in X_\P(\pos)$ is \emph{exposed in \pos} if $\pos_i = \{\ell\}$ for every $i \in [k]$ such that $\ell \in X_i(\pos)$. 
We say that position $\pos$ \emph{has an agreement tree} if $\P|X_\P(\pos)$ has an agreement tree.

A position \pos\ in \P\ is \emph{valid} if
$X_i(\pos_i) = X_\P(\pos) \cap X_i$,
for each $i \in [k]$.
The \emph{initial position} for $\P$ is the position $\initPos$, where, for each $i \in [k]$, $\initPos_i$ is a singleton set consisting of the label of $r(T_i)$ (i.e., $\initPos_i = \phi_i^{-1}(r(T_i))$.  Clearly, $\initPos$ is a valid position. 

\begin{lemma}\label{lemma:existenceofsupertree}
A profile $\P$ has an agreement tree if and only if there is an agreement tree for every valid position $\pos$ in $\P$.
\end{lemma}
%


\paragraph{Decomposing a position\label{sec:DecPos}.}

In what follows, $\pos$ denotes a valid position in $\P$.  For each $i \in [k]$ such that $\pos_i \neq \emptyset$, let $\ell_i \in X_i$ denote the single label in $\pos_i$.
Let  $\Ch_\P(\pos)$ denote the set of all children of some label in \pos; i.e., $\Ch_\P(\pos) = 
 \bigcup_{i \in [k]} \Ch_{T_i}(\pos_i)$.

Let $\pos$ be a valid position in $\P$.  
A \emph{good decomposition} of \pos\ is a pair $(S, \Pi)$, where $S$ is a subset of the exposed labels in $\bigcup_{i \in \pos_i} \pos_i$ and $\Pi = \{\pos^{(1)}, \pos^{(2)}, \dots, \pos^{(d)}\}$ is a collection of valid positions such that 
\begin{enumerate}[(D1)]
\item \label{itemD1}
$S \cup \bigcup_{j \in [d]} X_\P(\pos^{(j)}) = X_\P(\pos)$ and
$S \cap \bigcup_{j \in [d]} X_\P(\pos^{(j)}) = \emptyset$, and
\item  \label{itemD3}
$X_\P(\pos^{(p)}) \cap X_\P(\pos^{(q)}) = \emptyset$, for all $p, q \in [d]$ such that $p \neq q$.  
\end{enumerate}
Note that we allow  $S$ or $\Pi$ to be empty.  
We refer to the labels in $S$ as \emph{semi-universal labels} and to the positions in $\Pi$ as \emph{successor positions} of \pos.  The next result is central to our agreement algorithm.

\begin{lemma}\label{lemma:iffstatement}
Let $\pos$ be a valid  position in a profile $\P$.  
Then,
$\pos$ has an agreement tree
if and only if there exists a good decomposition $(S, \Pi)$ of $\pos$ such that $S \neq \emptyset$ and, for each position $\pi' \in \Pi$, $\pos'$ has an agreement tree.  If such a good decomposition exists, then $\pos$ has an agreement tree $\T = (T, \phi)$ where $\phi^{-1}(r(T)) = S$.
\end{lemma}

\paragraph{Good partitions.}
To find a good decomposition of a position $\pos$, it is convenient to work with partitions of $\Ch_\P(\pos)$. (Recall that a \emph{partition} of a set $Y$ is a collection $\Gamma$ of nonempty subsets of $Y$ such that every element $x \in Y$ is in exactly one set in $\Gamma$.)
A good decomposition $(S, \Pi)$, where $\Pi = \{\pos^{(j)}\}_{j \in [d]}$ defines a partition $\Gamma$ of the set $\Ch_\P(\pos)$ where, for any $a, b \in \Ch_\P(\pos)$, $a$ and $b$ are in the same set of $\Gamma$ if and only if there exists $j \in [d]$ such that $a, b \in X_\P(\pi^{(j)})$.  We refer to $\Gamma$ as the \emph{partition of $\Ch_\P(\pos)$ associated with $(S, \Pi)$}.  Next, we show that, conversely, certain partitions of $\Ch_\P(\pos)$ define good decompositions of \pos.


Set $A \subseteq \Ch_\P(\pos)$ is \emph{nice} with respect to a subset $S$ of the exposed labels in $\pos$ if, for each $i \in [k]$ and each label $\ell \in \bigcup_{i \in [k]} \pos_i$ such that $\Ch_\P(\ell) \cap A \neq \emptyset$,
\begin{enumerate}[(N1)]
\item 
if $\ell \in S$ and each $i \in [k]$ such that $\ell \in \pos_i$, then $|\Ch_{T_i}(\ell) \cap A| = 1$, and
\item 
if  $\ell \not\in S$, then $\Ch_\P(\ell) \subseteq X_\P(A)$. 
\end{enumerate}

Suppose $A$ is a nice set.  The \emph{position associated with $A$} is the position $\pos^A$, where, for each $i \in [k]$, $\pos_i^A$ is defined as follows.  If $\pos_i = \emptyset$, then $\pos_i^A = \emptyset$.  Otherwise, let $\ell$ be the single element in $\pos_i$.  Then,
\begin{equation}\label{eqn:posDefn}
\pos^A_i =
 \begin{cases}

\emptyset & \text{if $\Ch_{T_i}(\ell) \cap A = \emptyset$,} \\

\Ch_{T_i}(\ell) \cap A & \text{if $\ell \in S$, and} \\

\pos_i & \text{if $\ell \notin S$.}

 \end{cases}
 \end{equation}
A partition $\Gamma$ of $\Ch_\P(\pos)$ is \emph{good with respect to $S$} if each set $A \in \Gamma$ is nice with respect to $S$ and, for every two distinct sets $A, B \in \Gamma$, $X_\P(\pos^A) \cap X_\P(\pos^B) = \emptyset$.

\begin{lemma}\label{lemma:DecPart}
There is a bijection between good decompositions of $\pi$ and good partitions of $\Ch_\P(\pos)$.  That is, the following statements hold.
\begin{enumerate}[(i)]
\item Suppose $(S, \Pi)$ is a good decomposition of $\pos$.  Let $(S,\Gamma)$ be the partition of $\Ch_\P(\pos)$ associated with $(S,\Pi)$.  Then, $(S,\Gamma)$ is a a good partition of $\Ch_\P(\pos)$.
\item Suppose $(S,\Gamma)$ is a good partition of $\Ch_\P(\pos)$.  Let $\Pi = \{\pos^A: A \in \Gamma\}$.  Then, $(S,\Pi)$, a good decomposition of $\pos$.
\end{enumerate}
\end{lemma}

%

We refer to the good partition $(S,\Gamma)$ of $\Ch_\P(\pos)$ obtained from a good decomposition $(S,\Pi)$ of $\pos$, as described in Lemma \ref{lemma:DecPart} (i), as the \emph{good partition of $\Ch_\P(\pos)$ associated with  $(S,\Pi)$}. Likewise,
we refer to the good decomposition $(S,\Pi)$ of \pos\ obtained from a good partition $(S,\Gamma)$ of $\Ch_\P(\pos)$, as described in Lemma \ref{lemma:DecPart} (ii), as the \emph{good decomposition of $\Ch_\P(\pos)$ associated with $(S,\Gamma)$}.

Let  $(S,\Gamma), (S',\Gamma')$ be good partitions of $\Ch_\P(\pos)$. We say that $(S,\Gamma)$ is \textit{finer} than $(S',\Gamma')$, denoted $(S,\Gamma) \sqsubseteq (S',\Gamma')$, if and only if, $S \supseteq S'$ and, for every $A \in \Gamma$, there exists an $A' \in \Gamma'$ such that $A \subseteq A'$. We write $(S,\Gamma) \sqsubset (S',\Gamma')$ to denote that $(S,\Gamma) \sqsubseteq (S',\Gamma')$ and $(S,\Gamma) \neq (S',\Gamma')$.  
We say that a partition $(S,\Gamma)$ of $\Ch_\P(\pos)$ is \emph{minimal} if there does not exist another partition $(S',\Gamma')$ of $\Ch_\P(\pos)$ such that $(S',\Gamma') \sqsubset (S,\Gamma)$.

\begin{lemma}\label{lemma:MinimalPart}
Let  $\pos$ be a valid position in a profile \P.  Then, the minimal good partition of $\Ch_\P(\pos)$   is unique.
\end{lemma}


We refer to the (unique) good decomposition $(S,\Pi)$ associated 
with the minimal good partition of  $\Ch_\P(\pos)$ as the \emph{maximal good decomposition of $\pos$}.

\begin{corollary}\label{cor:nonemptySemi}
Let $\pos$ be a valid position in a profile $\P$ and $(S,\Pi)$ be the maximal good decomposition of $\pos$. If $\pos$ has an agreement tree, then $S \neq \emptyset$.
\end{corollary}


%
%
\section{Constructing an Agreement Tree\label{sec:constructingATs}}
%
 
\BuildAST (Algorithm \ref{alg:GenBuild}) takes as input a profile $\P$ on a set of labels $X$ and either returns an agreement tree for \P\ or reports that no such tree exists. \BuildAST assumes the availability of an algorithm \ComputeSuccessor that, given a valid position $\pos$ in $\P$, returns a maximal good decomposition $(S,\Pi)$ of $\pos$.

\begin{algorithm}[t]
\Fn(){\BuildAST{$\P$}}{
\SetAlgoLined
\SetNoFillComment
\DontPrintSemicolon
\SetAlgoLined
\KwData{A profile $\P = \{\T_1, \T_2, \dots , \T_k\}$ on a set of taxa $X$. }
\KwResult{Returns an agreement tree $\T$ for $\P$, if one exists; otherwise, returns \DISAGREEMENT.}
$\Q.\enqueue(\langle \initPos, \NULL\rangle)$ \label{algo:initQ}\;
\While{$\Q \neq \emptyset$\label{algo:beginWhile}}{
	$\langle \pos, \pred \rangle = Q.\dequeue()$\;
	$\langle S, \Pi \rangle = \ComputeSuccessor{$\pos$}$ \label{algo:successors}\;
	\If{$S= \emptyset$}{\Return \DISAGREEMENT\label{algo:absterminate}}
	Create a node $r(\pos)$\label{algo:createRoot}\;
	$r(\pos).\parent = \pred$\label{algo:setParent} \;
	\ForEach{$\ell \in S$\label{algo:forSetPhi}}{
		$\phi(\ell) = r(\pos)$\label{algo:setPhi}
	}
	\ForEach{$\pi' \in \Pi$\label{algo:forAddToQ}}{$Q.\enqueue(\langle \pos', r(\pos) \rangle)$\label{algo:addToQ}}
}
\Return $\T = (T, \phi)$, where $T$ is the tree with root $r(\initPos)$\label{algo:returnAST}\;
}
\vspace{\parsep}
\caption{Testing agreement\label{alg:GenBuild}}
\end{algorithm}

\BuildAST proceeds from the top down, starting from the initial position $\initPos$ of \P, attempting to construct an agreement tree for \P\ in a breath-first manner.  Like other algorithms based on breadth-first search, \BuildAST uses a queue, which stores pairs $\langle \pos, \pred \rangle$, where $\pos$ is a position in \P\ and \pred\ is a reference to the parent of the tree node (potentially) to be created for \pos. At the outset, the queue contains only the pair $\langle \initPos, \NULL \rangle$, corresponding to the root of the agreement tree, which has no parent.  

At each iteration of its outer \textbf{while} loop (lines \ref{algo:beginWhile}--\ref{algo:addToQ}),  \BuildAST extracts a pair $\langle \pos, \pred \rangle$ from its queue and invokes \ComputeSuccessor to obtain a maximal good decomposition $(S,\Pi)$ of $\pos$.  If $S = \emptyset$, then, by Corollary \ref{cor:nonemptySemi}, no agreement tree for \pos\ exists.  \BuildAST reports this fact (line \ref{algo:absterminate}) and terminates. 

If $S \neq \emptyset$, \BuildAST creates a tree node $r(\pos)$ for \pos; $r(\pos)$ is the tentative root for the agreement tree for \pos.  
By Lemma \ref{lemma:iffstatement}, if \pos\ has an agreement subtree, then it has an agreement tree where $\phi(\ell) = r(\pos)$. Lines \ref{algo:forSetPhi}--\ref{algo:setPhi} set up the mapping $\phi$ accordingly.  
Also by Lemma \ref{lemma:iffstatement}, if \pos\ has an agreement tree, then so does each position $\pos' \in \Pi$; furthermore, the roots of the trees for each position in $\Pi$ will be the children of $r(\pos)$.  Thus, \BuildAST adds $\langle \pos', r(\pos) \rangle$, for each $\pos' \in \Pi$ to the queue, to ensure that $\pos'$ is processed at a later iteration and that the root of the agreement tree constructed for $\pos '$ (if such a tree exists) is made to have $r(\pos)$ as its parent (lines \ref{algo:forAddToQ}--\ref{algo:addToQ}).  Therefore, if \BuildAST terminates without reporting \DISAGREEMENT, then the result returned in line \ref{algo:returnAST} is an agreement tree for \P.  \BuildAST indeed terminates, because there are only two possibilities at any given iteration: either the algorithm terminates reporting  \DISAGREEMENT\ or (since $S \neq \emptyset$) the maximal good decomposition $(S, \Pi)$ of \pos\ has the property that $\bigcup_{\pos' \in \Pi} X_\P(\pos')$ is a \emph{proper} subset of $X_\P(\pos)$. 
The number of iterations of \BuildAST cannot exceed the total number of nodes in an agreement tree for \P, which is $O(n)$.  Thus, we have the following result.

\begin{theorem}\label{theo:buildast}
Given a profile $\P = \{\T_1, \T_2, \dots, \T_k\}$,
\BuildAST returns an agreement tree $\T$ for $\P$, if such a tree exists; otherwise, \BuildAST returns \DISAGREEMENT.  The total number of iterations of \BuildAST's outer loop is $O(n)$.
\end{theorem}



\paragraph{Finding the maximal good decomposition\label{sec:findSuccessors}.}



\begin{algorithm}[t]
\Fn(){\ComputeSuccessor{$\pi$}}{
\SetAlgoLined
\SetNoFillComment
\DontPrintSemicolon
\SetAlgoLined
\KwData{A valid position $\pi$.}
\KwResult{Returns the maximal good decomposition $(S,\Pi)$ of $\pos$.}
$S = \{\ell: \ell \text{ is exposed in \pos}\}$, $K = \{i: \pos_i = \{\ell\} \text{ for some } \ell \in S\}$ \label{algo:initSK}\;
$\Gamma = \{A:  A = W \cap \Ch_\P(\pos), \text{ for some connected component $W$ of $\DG(\pos) \setminus S$}\}$ \label{algo:initGamma}\;
\While{$S$ contains a bad label\label{algo:getSuccWhile}}{
	Choose any bad label $\ell \in S$ \label{algo:chooseEll} \;
	$K' = \{i: \pos_i = \{\ell\}\}$\label{algo:retrieveKp}\;
	$\Gamma' = \{\CS \in \Gamma:  \Ch_{T_i}(\ell) \cap \CS \neq \emptyset \text{ for some } i \in K'\}$\label{algo:buildGammap} \; 
	$B = \bigcup_{\CS \in \Gamma'} \CS$\label{algo:unions}\;
	$\Gamma = \Gamma \setminus \Gamma' \cup \{B\}$\label{algo:updateGamma} \; 
	$S = S \setminus \{\ell\}$, $K = K \setminus K'$\;  \label{algo:getSuccWhileEnd}
}
$\Pi \leftarrow \emptyset$ \label{algo:initPi}\;
\ForEach{$\CS \in \Gamma$}{\label{algo:foreach}
	\lForEach{$i \in [k]$}{$\pos^\CS_i = \emptyset$}
	\ForEach{$i \in [k]$\label{algo:innerTest}}{
		Let $\ell$ be the single label in $\pos_i$ \;
		\If{$\Ch_{T_i}(\ell) \cap \CS \neq \emptyset$}{
			\lIf{$\ell \in S$}{
				$\pos^\CS_i = \Ch_{T_i}(\ell) \cap \CS$
			}
			\lElse{$\pos^\CS_i =  \pos_i$}
		}
	}
	$\Pi = \Pi \cup \pos^\CS$ \label{algo:updatePi}\;
}
\Return $(S,\Pi)$\; 
}
\caption{Computing the maximal good decomposition.\label{algo:abscomputesuccessor}}
\end{algorithm}

\ComputeSuccessor (Algorithm \ref{algo:abscomputesuccessor}) computes a maximal good decomposition of a position $\pos$, relying on 
an auxiliary graph known as \emph{the display graph} of the input profile and denoted by $\DG$ \cite{BryantLagergren06,DengFB2017b,DengFB2017}.
The graph $\DG$ is obtained from the disjoint union of the underlying trees $T_1, \dots, T_k$ of the $\P$ by identifying nodes that have the same label. Multiple edges between the same pair of nodes are replaced by a single edge. See Figure~\ref{fig:displayGraph}. 

\begin{figure}[t]
\centering
  \includegraphics[scale=0.45]{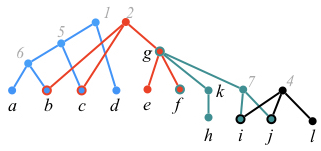}
  \caption{The display graph $\DG$ for the profile of Figure \ref{fig:profile}.} 
  \label{fig:displayGraph}
\end{figure}

$\DG$ has $O(nk)$ nodes and edges, and can be constructed in $O(nk)$ time.  By Assumption \ref{as:fullySingular}, %
there is a bijection between the labels in $X$ and the nodes of $\DG$.  Thus, from this point forward, we refer to the nodes of $\DG$ by their labels.
%
For a valid position  $\pos$, $\DG(\pos)$ denotes the subgraph of $\DG$ induced by $X(\pos)$. Thus, $\DG(\initPos) = \DG$.

%

Lines \ref{algo:initSK}--\ref{algo:getSuccWhileEnd} of \ComputeSuccessor construct the minimal good partition of $\Ch_\P(\pos)$. 
Line \ref{algo:initSK} initializes $S$ to contain all exposed labels in \pos, and sets $K$ to consist of the indices of the trees in \P\ that contain the labels in $S$.  Line \ref{algo:initGamma} initializes $\Gamma$ using $\DG(\pos)$.  
We say that a label $\ell \in S$ is \emph{bad} if there exist $i \in K$ and $\CS \in \Gamma$ such that $\pi_i = \{\ell\}$ and $|\Ch_{T_i}(\ell) \cap \CS| \ge 2$.
Lines \ref{algo:getSuccWhile}--\ref{algo:getSuccWhileEnd} construct the minimal nice partition $(S,\Gamma)$ of $\Ch_\P(\pos)$ by deleting bad labels from $S$ and merging sets in $\Gamma$ accordingly. 
Let $(S^*,\Gamma^*)$ denote the minimal good partition of $\Ch_\P(\pos)$.  

\begin{lemma}\label{lemma:computesuccessor}
Let \pos\ be a valid position in a profile $\P$ and let $(S^*,\Gamma^*)$ be the minimal good partition of $\Ch_\P(\pos)$.
Let $(S_0,\Gamma_0)$ denote the initial value of $(S,\Gamma)$ in \ComputeSuccessor before entering the \textbf{while} loop,
$(S_j,\Gamma_j)$ denote the value of $(S,\Gamma)$ after $j$ executions of the body of the loop, and $r$ denote the total number of iterations.  Then, 
$r \le k$ and
$(S_0,\Gamma_0) \sqsubset (S_1,\Gamma_1)  \sqsubset (S_2, \Gamma_2)  \sqsubset \dots  \sqsubset (S_r,\Gamma_r) = (S^*,\Gamma^*)$.  
\end{lemma}

By Lemma \ref{lemma:computesuccessor}, the pair $(S,\Gamma)$ constructed in lines  \ref{algo:getSuccWhile}--\ref{algo:getSuccWhileEnd} of \ComputeSuccessor is a minimal good partition of $\Ch_\P(\pos)$.
The \textbf{foreach} loop of lines \ref{algo:initPi}--\ref{algo:updatePi} simply uses Equation \eqref{eqn:posDefn} to construct the maximal good decomposition $(S,\Pi)$ of $\pos$ from $(S,\Gamma)$.  We thus have the following.

\begin{lemma}\label{lemma:getDec}
\ComputeSuccessor returns the maximal good decomposition of $\pos$. 
\end{lemma}

%

%
%
\paragraph{Implementation.}

Throughout its execution, \BuildAST maintains the display graph $\DG$. 
%
Also, for each label $\ell \in X$, it maintains a field $\ell.\expose$ containing every index $i$ such that $\pos_i = \{\ell\}$ for some $\pos$ in $Q$. Label $\ell$ is exposed when $|\ell.\expose| = k_\ell$, where $k_\ell$ denotes the number of input trees containing label $\ell$.
For each $\pos$ in \BuildAST's queue, the set $\Ch_\P(\pos)$ is stored as a sparse array $((i,\Ch_{T_i}(\pi_i)): i \in [k] \text{ and } \Ch_{T_i}(\pi_i)) \neq \emptyset)$.  This enables \ComputeSuccessor to access the parts of $\Ch_\P(\pos)$ associated with each input tree separately.  
We use this representation of $\Ch_\P(\pos)$ to build similar representations of the sets in the partition $\Gamma$ of $\Ch_\P(\pos)$ produced from $\DG(\pos) \setminus S$ in line \ref{algo:initGamma} of \ComputeSuccessor.
For each label $a \in \Ch_{\P}(\pos)$, we maintain a mapping that returns, in $O(1)$ time, the set $A \in \Gamma$ containing $a$. 
During the execution of \ComputeSuccessor's \textbf{while} loop, sets in $\Gamma$ may be merged, and representations of these merged sets must be produced and the mapping from $\Ch_{\P}(\pos)$ to $\Gamma$ must be modified.

\begin{lemma}\label{time:maintainGraph}
The total time needed to maintain the display graph throughout the entire execution of \BuildAST is $\mathcal{O}(nk\log^2 (nk))$.
\end{lemma}

In the following results, $d_i$ denotes the maximum number of children of a node in tree $T_i$, for each $i \in [d]$. 

\begin{lemma}\label{time:initGD}
Excluding the time needed to maintain the display graph, Lines \ref{algo:initSK} and \ref{algo:initGamma} of \ComputeSuccessor take $\mathcal{O}(\sum_{i \in [k]} d_i)$ time.
\end{lemma}

%

\begin{lemma}\label{time:getDec}
\ComputeSuccessor's \textbf{while} loop takes $\mathcal{O}(k\sum_{i \in [k]} d_i)$ time.  
\end{lemma}

\begin{theorem}\label{thm:totalTime}
\BuildAST can be implemented to run in $\Oh(n k (\sum_{i \in [k]} d_i + \log^2(nk)))$  time, where $n$ is the number of distinct taxa in \P, $k$ is the number of trees in \P, and $d_i$ is the maximum number of children of tree $T_i$, for $i \in [k]$.
\end{theorem}

\section{Concluding Remarks\label{sec:discussion}}

\BuildAST may be much faster in practice than Theorem \ref{thm:totalTime} suggests, since that bound assumes the unlikely scenario where every edge deletion performed in constructing $\DG(\pos) \setminus S$ in \ComputeSuccessor generates a new component and that most of these components are remerged in the \ComputeSuccessor's \textbf{while} loop.
In any case, Theorem \ref{thm:totalTime} implies that \BuildAST performs well if the sum of the maximum out-degrees is small relative to the number of taxa.

The running time of $\BuildAST$ can be further improved to $\Oh(n k (\sum_{i \in [k]} d_i + \log^2(nk)/\log \log (nk)))$ using the graph connectivity data structure of \cite{Wulff-Nilsen2013}.  It is not clear, however, that this data structure would have a practical impact.  In fact, experimental work \cite{FBLiu2019} suggests that data structures much simpler than HDT (and, therefore, than \cite{Wulff-Nilsen2013}) perform well in practice.  

\BuildAST can be modified to run in $\mathcal{O}(nk \log^2 (nk))$ time for profiles $\P$  where the input trees are all binary and solely leaf-labeled. For such profiles, $|A \cap \Ch_{T_i}(\pos_i)| \leq 2$, for $A \in \Gamma$ and $i \in [k]$ in a position $\pos$ of $\P$. Labels $a, a' \in \Ch_{T_i}(\pos_i)$ are either in the same set $A$ or in different sets $A, A'$ where $A, A' \in \Gamma$.  In the first case, $\ell \in \pos_i$ must be bad. Bad labels can then be detected earlier in Line \ref{algo:initGamma} and directly removed from $S$. Thus, we can skip \ComputeSuccessor's \textbf{while} loop. Hence, maintaining graph connectivity dominates the performance of \BuildAST. 

\BuildAST enables users to deal with hard polytomies.  In applications, we may encounter both hard and soft polytomies.  It would be interesting to modify \BuildAST to handle a mixture of both types polytomies, as appropriate.

\bibliographystyle{splncs04}
\bibliography{../../Bibliographies/mybib,../../Bibliographies/phylogenies}

\begin{thebibliography}{10}
\providecommand{\url}[1]{\texttt{#1}}
\providecommand{\urlprefix}{URL }
\providecommand{\doi}[1]{https://doi.org/#1}

\bibitem{AhoSagivSzymanskiUllman81}
Aho, A., Sagiv, Y., Szymanski, T., Ullman, J.: Inferring a tree from lowest
  common ancestors with an application to the optimization of relational
  expressions. {SIAM J. Computing}  \textbf{10}(3),  405---421 (1981)

\bibitem{Baum:1992}
Baum, B.R.: Combining trees as a way of combining data sets for phylogenetic
  inference, and the desirability of combining gene trees. Taxon  \textbf{41},
  3--10 (1992)

\bibitem{BerrySemple2006}
Berry, V., Semple, C.: Fast computation of supertrees for compatible
  phylogenies with nested taxa. Systematic Biology  \textbf{55}(2),  270--288
  (2006)

\bibitem{BinindaEmonds04}
Bininda-Emonds, O.R.P. (ed.): Phylogenetic Supertrees: Combining Information to
  Reveal the Tree of Life, Series on Computational Biology, vol.~4. Springer,
  Berlin (2004)

\bibitem{BordewichEvansSemple2006}
Bordewich, M., Evans, G., Semple, C.: Extending the limits of supertree
  methods. Annals of Combinatorics  \textbf{10},  31--51 (2006)

\bibitem{BryantLagergren06}
Bryant, D., Lagergren, J.: Compatibility of unrooted phylogenetic trees is
  {FPT}. Theoretical Computer Science  \textbf{351},  296--302 (2006)

\bibitem{DanielSemple2004}
Daniel, P., Semple, C.: Supertree algorithms for nested taxa. In:
  Bininda-Emonds, O.R.P. (ed.) Phylogenetic supertrees: Combining information
  to reveal the Tree of Life, pp. 151--171. Kluwer, Dordrecht (2004)

\bibitem{DengFB2017b}
Deng, Y., Fern\'andez-Baca, D.: An efficient algorithm for testing the
  compatibility of phylogenies with nested taxa. Algorithms for Molecular
  Biology  \textbf{12}, ~7 (2017)

\bibitem{DengFB2017}
Deng, Y., Fern\'andez-Baca, D.: Fast compatibility testing for rooted
  phylogenetic trees. Algorithmica  \textbf{80}(8),  2453--2477 (2018).
  \doi{10.1007/s00453-017-0330-4}, \url{http://rdcu.be/thB1}

\bibitem{FBGSV2015}
Fern{\'a}ndez-Baca, D., Guillemot, S., Shutters, B., Vakati, S.:
  Fixed-parameter algorithms for finding agreement supertrees. SIAM Journal on
  Computing  \textbf{44}(2),  384--410 (2015)

\bibitem{FBLiu2019}
Fern{\'a}ndez-Baca, D., Liu, L.: Tree compatibility, incomplete directed
  perfect phylogeny, and dynamic graph connectivity: An experimental study.
  Algorithms  \textbf{12}(3), ~53 (2019)

\bibitem{HinchliffPNAS2015}
Hinchliff, C.E., Smith, S.A., Allman, J.F., Burleigh, J.G., Chaudhary, R.,
  Coghill, L.M., Crandall, K.A., Deng, J., Drew, B.T., Gazis, R., Gude, K.,
  Hibbett, D.S., Katz, L.A., Laughinghouse~IV, H.D., McTavish, E.J., Midford,
  P.E., Owen, C.L., Reed, R.H., Reesk, J.A., Soltis, D.E., Williams, T.,
  Cranston, K.A.: Synthesis of phylogeny and taxonomy into a comprehensive tree
  of life. Proceedings of the National Academy of Sciences  \textbf{112}(41),
  12764--12769 (2015). \doi{10.1073/pnas.1423041112}

\bibitem{HolmLichtenbergThorup:2001}
Holm, J., de~Lichtenberg, K., Thorup, M.: Poly-logarithmic deterministic
  fully-dynamic algorithms for connectivity, minimum spanning tree, 2-edge, and
  biconnectivity. J. ACM  \textbf{48}(4),  723--760 (Jul 2001).
  \doi{10.1145/502090.502095}, \url{http://doi.acm.org/10.1145/502090.502095}

\bibitem{JanssonRECOMB2017}
Jansson, J., Lingas, A., Rajaby, R., Sung, W.K.: Determining the consistency of
  resolved triplets and fan triplets. In: International Conference on Research
  in Computational Molecular Biology. pp. 82--98. Springer (2017)

\bibitem{Maddison89}
Maddison, W.P.: Reconstructing character evolution on polytomous cladograms.
  Cladistics  \textbf{5},  365--377 (1989)

\bibitem{NgWormald96}
Ng, M., Wormald, N.: Reconstruction of rooted trees from subtrees. Discrete
  Applied Mathematics  \textbf{69}(1--2),  19--31 (1996)

\bibitem{Page2004}
Page, R.M.: Taxonomy, supertrees, and the tree of life. In: Bininda-Emonds,
  O.R.P. (ed.) Phylogenetic supertrees: Combining information to reveal the
  Tree of Life, pp. 247--265. Kluwer, Dordrecht (2004)

\bibitem{Ragan:1992}
Ragan, M.A.: Phylogenetic inference based on matrix representation of trees.
  Molecular Phylogenetics and Evolution  \textbf{1},  53--58 (1992)

\bibitem{RedelingsHolder2017}
Redelings, B.D., Holder, M.T.: A supertree pipeline for summarizing
  phylogenetic and taxonomic information for millions of species. PeerJ
  \textbf{5},  e3058 (2017). \doi{10.7717/peerj.3058}

\bibitem{Sanderson:2008}
Sanderson, M.J.: Phylogenetic signal in the eukaryotic tree of life. Science
  \textbf{321}(5885),  121--123 (2008)

\bibitem{NCBI2009}
{Sayers et al.}, E.W.: Database resources of the {National Center for
  Biotechnology Information}. Nucleic Acids Research  \textbf{37(Database
  issue)},  D5--D15 (2009)

\bibitem{SempleSteel03}
Semple, C., Steel, M.: Phylogenetics. Oxford Lecture Series in Mathematics,
  Oxford University Press, Oxford (2003)

\bibitem{APG2016}
{The Angiosperm Phylogeny Group}: An update of the {Angiosperm Phylogeny Group}
  classification for the orders and families of flowering plants: {APG IV}.
  Botanical Journal of the Linnean Society  \textbf{181},  1--20 (2016)

\bibitem{WarnowSupertrees2018}
Warnow, T.: Supertree construction: Opportunities and challenges. Tech. Rep.
  arXiv:1805.03530, ArXiV (May 2018), \url{https://arxiv.org/abs/1805.03530}

\bibitem{Wulff-Nilsen2013}
Wulff-Nilsen, C.: Faster deterministic fully-dynamic graph connectivity. In:
  Proceedings of the Twenty-fourth Annual ACM-SIAM Symposium on Discrete
  Algorithms. pp. 1757--1769. SODA '13, Society for Industrial and Applied
  Mathematics, Philadelphia, PA, USA (2013),
  \url{http://dl.acm.org/citation.cfm?id=2627817.2627943}

\end{thebibliography}

\newpage

\appendix

\section{Omitted Proofs}

\subsection*{Proof of Lemma \ref{lemma:profileEmbedding}}

Suppose there is a node $v \in V(T)$ such that $\phi^{-1}(v) = \emptyset$.  Note that $v$ cannot be a leaf. Let $u_1, u_2, \dots , u_d$ be the children of $v$.  One can prove the following.

\begin{quote}\textbf{Fact.}
\emph{For each $i \in [k]$, there is at most one $j \in [d]$ such that $X_\P(u_j) \cap X_i \neq \emptyset$.}
\end{quote}

Now, choose any $j \in [d]$.  Let $T'$ be the tree obtained by contracting the edge $(v,u_j) \in E(T)$.  That is, $T'$ is obtained by eliminating edge $(v,u_j)$, deleting $u_j$, and making $\Ch_{T'}(v) = \Ch_T(v) \cup \Ch_T(u_j)$.  Let $\T' = (T', \phi')$, where $(\phi')^{-1}(w) = \phi^{-1}(w)$, if $w \in V(T) \setminus \{v,u_j\}$ and $(\phi')^{-1}(v) = \phi^{-1}(v) \cup \phi^{-1}(u_j)$.  Then, the above fact implies that, for each $i \in [k]$, $\Cl(\T|X_i) = \Cl(\T'|X_i)$. That is, $\T'$ is also an agreement tree for $\P$.  If we repeat this contraction operation until it no longer applies, the final tree $\T'' = (T'',\phi'')$ will satisfy $(\phi'')^{-1}(v) \neq \emptyset$ for each node $v \in V(T'')$.
\QED

\subsection*{Proof of Lemma \ref{lemma:embedding}}

We argue that, for each $i \in [k]$, $\T$ agrees with $\T_i$ if and only if there exists a topological embedding $\phi_i$ of $\T_i$ into $\T$.

($\Longrightarrow$) Suppose that $\phi_i$ is a topological embedding from $\T_i$ to $\T$. We show that $\Cl(\T_i) = \Cl(\T|X_i)$, which implies that $\T$ agrees with $\T_i$. 

First, we show that $\Cl(\T_i) \subseteq \Cl(\T|X_i)$ by arguing that, for each $a \in X_i$, $X_i(a) = X_\P(\phi(a)) \cap X_i$. By definition of LCAs, $X_i(a) \subseteq  X_\P(\phi(a))$.  Now, suppose that there is a label $b \in X_\P(\phi(a)) \cap X_i$ such that $b \notin X_i(a)$.  Let $c = \LCA_{T_i}(X_i(a) \cup \{b\})$.  Then, $c <_{T_i} a$.  On the other hand, $\phi(c) \ge_T \phi(a)$, contradicting condition (E2).

Next, we prove that  $\Cl(\T|X_i) \subseteq \Cl(\T_i)$. Suppose, to the contrary, that there is a cluster $Y \in \Cl(\T|X_i) \setminus \Cl(\T_i)$.
Let $a \in X_i$ be the (unique) label such that $X_i(a) \supset Y$ and for every $b \in \Ch_{T_i}(a)$ either $X_i(b) \subset Y$ or $X_i(b) \cap Y = \emptyset$.  Since $Y \notin \Cl(\T_i)$, there must exist at least two labels $c_1, c_2 \in \Ch_{T_i}(a)$ such that $X_i(c_1), X_i(c_2) \subset Y$ and at least one label $c_3 \in \Ch_{T_i}(a)$ such that $X_i(c_3) \cap Y = \emptyset$. Therefore, there is a single node $v \in \Ch_T(\phi_i(a))$ such that $\phi_i(c_1), \phi_i(c_2) \in X_\P(v)$, contradicting condition (E3).

($\Longleftarrow$) Suppose that $\T$ agrees with $\T_i$.  It is straightforward to show that $\phi_i$ must satisfy (E1).  Thus, we focus on conditions (E2) and (E3).

Suppose condition (E2) does not hold.  Then, there exists a label $b \in \Ch_{T_i}(a)$, such that  $\phi_i(a) \ge_{T} \phi_i(b)$.  Since $X_\P(b) \subset X_\P(a)$, we must in fact have $\phi_i(a) >_{T} \phi_i(b)$. But then $\T$ does not agree with $\T_i$, a contradiction.   

Suppose condition (E3) does not hold.  Then there are distinct labels $b, c \in \Ch_{T_i}(a)$ such that $\{\phi_i(b), \phi_i(c)\} \subseteq X_\P(v)$, for some $v \in \Ch_T(\phi_i(a))$. But then $\T|X_i$ contains a cluster not in $\T_i$, contradicting the assumption that $\T$ agrees with $\T_i$.
\QED

\subsection*{Proof of Lemma \ref{lemma:existenceofsupertree}}

$(\Longrightarrow)$ Suppose $\P$ has an agreement tree $\T$. For any valid position $\pos$ in $\P$, $X_\P(\pos) \subseteq X_\P$.  Thus, by Lemma \ref{lm:compatSubprofile}, $\T|X_\P(\pos)$ is an agreement tree for $\pos$. 

\noindent
($\Longleftarrow$) Suppose there is an agreement tree for every valid position $\pos$ in $\P$.  Then there must exist an agreement tree $\T$ for the initial position \initPos\ of $\P$. Since $X_\P(\initPos) = X_\P$, $\T$ must also be an agreement tree for $\P$.
\QED

\subsection*{Proof of Lemma \ref{lemma:iffstatement}}


$(\Longrightarrow)$ Suppose position \pos\ has an agreement tree $\T = (T, \phi)$ (thus, $\T$ is an $X_\P(\pos)$-tree). If $T$ consists of a single node $u = r(T)$, then we must have $\phi^{-1}(u) = \bigcup_{i \in [k]} \pos_i$. Clearly, every label in $S$ is exposed.  Let $(S, \Pi) = (\phi^{-1}(u), \emptyset)$. Since $\Pi = \emptyset$, $S = \bigcup_{i \in [k]} \pos_i = X_\P(\pos)$, so (D\ref{itemD1}) holds, and condition (D\ref{itemD3}) holds trivially.  Thus, $(S, \Pi)$ is a good decomposition of $\pos$. By Lemma \ref{lemma:profileEmbedding}, $S \neq \emptyset$.  

Now, suppose $\Ch_T(r(T)) = \{v_1, v_2,\dots, v_d\}$, where $d \ge 1$.
By Lemma \ref{lm:compatSubprofile}, for each $j \in [d]$, $\T|X_\P(v_j)$ is an agreement tree for $\P|X_\P(v_j)$.
For each $i \in [k]$ and each $j \in [d]$ such that $X_i \cap X_\P(v_j) \neq \emptyset$, let $\ell^{(j)}_i$ denote the label of the root of $\T_i|\left(X_i \cap X_\P(v_j)\right)$.  For each $j \in [d]$, define a position $\pos^{(j)}$, where, for each $i \in [k]$,
\[
\pos^{(j)}_i =
\begin{cases}
\emptyset & \text{if 
$X_i \cap X_\P(v_j) = \emptyset$} \\
\left \{\ell^{(j)}_i \right \} & \text{otherwise.}
\end{cases}
\]

Let  $S = \phi^{-1}(r(T))$ and $\Pi =  \{\pos^{(1)}, \pos^{(2)}, \dots , \pos^{(d)}\}$.  By Lemma \ref{lemma:profileEmbedding}, we can assume that $S \neq \emptyset$.
It is straightforward to show that each label in $S$ is exposed in \pos, and that, for each $j \in [d]$, position $\pos^{(j)}$ is valid.  It can also be shown that the pair $(S, \Pi)$ satisfies properties (D\ref{itemD1}) and (D\ref{itemD3}).  Thus, $(S, \Pi)$ is a good decomposition of $\pos$.

$(\Longleftarrow)$ 
Let  $(S, \Pi)$ be a good decomposition of $\pos$ such that $S \neq \emptyset$ and each position in $\Pi$ has an agreement tree.   If $\Pi = \emptyset$, then we must have $S = X_\P(\pos)$.  Let $T$ be the tree consisting of a single node $u = r(T)$ and let $\phi(\ell) = u$, for all $u \in S$.  Then, $\T = (T,\phi)$ is an agreement tree for $\pos$.

Now suppose $\Pi \neq \emptyset$.  Let $\posSet = \{\pos^{(1)}, \pos^{(2)}, \dots , \pos^{(d)}\}$.  For each $j \in [d]$, let $\T^{(j)} = (T^{(j)}, \phi^{(j)})$ be an agreement tree for $\pos^{(j)}$, and let $v_j$ be the root of $T^{(j)}$.   Let $\T = (T, \phi)$ be the $X_\P(\pos)$-tree where $T$ is assembled by creating a new node $u$ and making $\Ch_T(u) = \{v_1, v_2, \dots , v_d\}$ and, for each $\ell \in X_\P(\pi)$, $\phi(\ell)$ is defined as
$$
\phi(\ell) =
\begin{cases}
u & \text{if $\ell \in S$} \\
\phi^{(j)}(\ell) & \text{if $\ell \in X_\P(\pos^{(j)})$.}
\end{cases}
$$
Note that conditions  (D\ref{itemD1}) and (D\ref{itemD3}) imply that $\T$ is indeed an $X_\P(\pos)$-tree.
We claim that $\T$ is an agreement tree for $\pos$. By Lemma \ref{lemma:embedding}, it suffices to show that, for each $i \in [k]$, $\phi$ is a topological embedding from $\T_i|X_i(\pos_i)$ to $\T$.  
Lemma \ref{lemma:embedding} implies that, for each $i \in [k]$, $\phi^{(j)}$  is a topological embedding from $\T_i| X_i (\pos_i^{(j)})$ to $\T^{(j)}$.  Thus, every node $u$ in $\T_i|X_i(\pos_i^{(j)})$ satisfies (E1)--(E3).  For each $j \in [d]$, let $\ell_i$ be the label of the root of $\T_i|X_i(\pos_i^{(j)})$.
There are two possibilities:
\begin{enumerate}
\item $\ell_i \in \phi^{-1}(u)$.  Then, each of $\ell_i$'s children must be in a distinct subtree of $u$.  Thus, properties (E1)--(E3) are satisfied.
\item
 $\ell_i \not\in \phi^{-1}(u)$.  Then, $\ell_i$ and all of its children must be contained in a single subtree, say $\T_j$, of $u$, and the claim follows from the fact that $\phi^{(j)}$  is a topological embedding.
\end{enumerate}
\QED

\subsection*{Proof of Lemma \ref{lemma:DecPart}}

\begin{enumerate}[(i)]
\item 
For each $\ell \in \bigcup_{i \in [k]} \pos_i$, the following statements hold.
\begin{enumerate}[(a)]
\item 
If $\ell \in S$ and each $i \in [k]$ such that $\ell \in \pos_i$, then each label in $\Ch_{T_i}(\ell)$ is in a distinct subset of $\Gamma$.
\item 
If $\ell \not\in S$, then there exists a set $A \in \Gamma$ such that $\Ch_\P(\ell) \subseteq X_\P(A)$. 
\end{enumerate}
Thus, (N1) and (N2) hold.  Since $(S,\Pi)$ is a good decomposition, we also have that $X_\P(\pos^A) \cap X_\P(\pos^B) \ \emptyset$, for pair $A, B$ of distinct sets in $\Gamma$.  Hence, $(S,\Gamma)$ is a good partition of $\Ch_\P(\pos)$.
\item
Note that, each $\pos' \in \Gamma$ is valid, that (D1) holds by construction, and that  (D2) holds by definition.  Therefore, $(S,\Pi)$ is a good decomposition of \pos. 

\QED
\end{enumerate}

\subsection*{Proof of Lemma \ref{lemma:MinimalPart}}

In order to prove this lemma, we need to introduce a new concept.  Let $(S,\Gamma)$ and $(S',\Gamma')$ be two partitions of $\Ch_\P(\pos)$. We write $(S,\Gamma)\sqcap (S',\Gamma')$ to denote the partition $(S'',\Gamma'')$ where $S'' = S \cup S'$ and $\Gamma'' = \{A \cap B : A \in \Gamma, B \in \Gamma'\} \setminus \{\emptyset\}$. 

\begin{lemma}\label{lemma:nicePartitionOnSqSet}
Let $(S,\Gamma)$ and $(S',\Gamma')$ be two good partitions of $\Ch_\P(\pos)$, and let $(S'',\Gamma'') = (S,\Gamma) \sqcap (S',\Gamma')$. Then, $(S'',\Gamma'')$ is also a good partition of $\Ch_\P(\pos)$.
\end{lemma}

\begin{proof}
We first show that for any distinct $A, B \in \Gamma''$, $X_\P(\pos^A) \cap X_\P(\pos^B) = \emptyset$. We have $A = C \cap C'$, for some $C \in \Gamma$, $C' \in \Gamma'$, and $B = D \cap D'$, for some $D \in \Gamma$, $D' \in \Gamma'$. Since $A \neq B$, we have $C \neq D$ or $C' \neq D'$.  In the first case, $X_\P(\pos^{C}) \cap X_\P(\pos^{D}) = \emptyset$, and, in the second case, $X_\P(\pos^{C'}) \cap X_\P(\pos^{D'}) = \emptyset$. Thus,  $X_\P(\pos^A) \cap X_\P(\pos^B) = \emptyset$.

Next, we show that each set $A \in \Gamma''$ is nice with respect to $S''$.  Suppose $A = C \cap D$, where $C \in \Gamma$ and $D \in \Gamma'$. Consider any $\ell \in \bigcup_{i \in [k]} \pos_i$ and each $i \in [k]$ such that $\Ch_\P(\ell) \cap A \neq \emptyset$.  It must be the case that $\Ch_\P(\ell) \cap C \neq \emptyset$ and $\Ch_\P(\ell) \cap D \neq \emptyset$.  Suppose $\ell \in S''$.  Then, either $\ell \in S$ or $\ell \in S'$, and (N1) must hold for $A$.  If $\ell \notin S''$, then $\ell \notin S$ and $\ell \notin S$.  Thus, by (N2), $\Ch_\P(\ell) \subset C$ and $\Ch_\P(\ell) \subseteq D$ and, thus $\Ch_\P(A) \subseteq D$.  Hence, (N2) holds for $A$. \QED
\end{proof}

Now, to prove Lemma \ref{lemma:MinimalPart}, suppose, on the contrary, that there exist at least two distinct minimal good partitions $(S,\Gamma), (S',\Gamma')$. By Lemma \ref{lemma:nicePartitionOnSqSet}, $(S'', \Gamma'') = (S,\Gamma) \sqcap (S',\Gamma')$ is also a good partition of $\Ch_\P(\pos)$. But $(S'',\Gamma'') \sqsubset (S,\Gamma)$, contradicting the assumption that $(S,\Gamma)$ is minimal. \QED


\subsection*{Proof of Corollary \ref{cor:nonemptySemi}}

Suppose, on the contrary, that $\pos$ has an agreement tree, but $S = \emptyset$.   Since $(S,\Pi)$ is minimal, Lemma \ref{lemma:MinimalPart} implies that $S' = \emptyset$ for \emph{every} good decomposition $(S', \Pi')$ of $\pos$. But, by Lemma \ref{lemma:iffstatement}, this implies that $\pos$ has no agreement tree, a contradiction.
\QED

\subsection*{Proof of Lemma \ref{lemma:computesuccessor}}

The $j$th execution of the body of the loop, $j > 1$, removes one bad label from $S_{j-1}$. Thus, $S_j \subset S_{j-1}$.  Since $|S| \le k$, this implies that the number of iterations is at most $k$.

Every set $\CS' \in \Gamma_j$ is either in $\Gamma_{j-1}$ or is the union of two or more sets in $\Gamma_{j-1}$. Hence, for every set $\CS \in {\Gamma}_{j -1}$, there exists a set $\CS' \in {\Gamma}_j$ such that $\CS \subseteq \CS'$.  Thus, $(S_{j-1}, \Gamma_{j-1}) \sqsubset (S_j, \Gamma_j)$.


Next, we argue that $(S_j, \Gamma_j) \sqsubseteq (S^*, \Gamma^*)$, for each $j \in \{0,1, \dots , r\}$. Consider $j = 0$.  
We use the following observation.

\begin{quote}
For $i \in [2]$, let $R_i$ be a a subset of $\bigcup_{i \in [k]} \pos_i$ and let
$\Gamma_i = \{A: A = W \cap \Ch_\P(\pos),$ for some connected component $W$ of $\DG(\pos) \setminus R_i\}$.  If $R_1 \subseteq R_2$, then, for each set $A \in \Gamma_1$, there exists a set $B \in \Gamma_2$ such that $A \supseteq B$.
\end{quote}

The above observation and the fact that $S_0 \supseteq S^*$ imply that for each set $A \in \Gamma_0$, there exists a set $A^* \in \Gamma^*$ such that $A \subseteq A^*$.  Thus, $(S_0,\Gamma_0) \sqsubseteq (S^*,\Gamma^*)$.

Now assume that $(S_p, \Gamma_p) \sqsubseteq (S^*, \Gamma^*)$, for each $p \in \{0,1, \dots , j-1\}$, $j > 1$.  By the observation above, it suffices to show that $S_j \supseteq S^*$. Note that $S_j = S_{j-1} \setminus \{\ell\}$, where $\ell$ is the bad label chosen in line \ref{algo:chooseEll}, which cannot be in $S^*$.  Thus, $S_j \supseteq S^*$.

We claim that, for each $j \in \{0,1, \dots , r\}$, each $\ell \in \bigcup_{i \in [k]} \pos_i \setminus S_j$, there is an $A \in \Gamma_j$ such that $\Ch_\P(\ell) \subseteq X_\P(A)$.  This is true by construction for $j = 0$, and the body of the \textbf{while} loop ensures that this remains true throughout the execution of the algorithm.
At termination of the \textbf{while} loop, $S_r$ contains no bad labels. Thus, 
$(S_r,\Gamma_r)$ satisfies properties (N1) and (N2). Further, it can be shown that for every two distinct sets $A,B \in \Gamma$, $X_\P(\pos^A) \cap X_\P(\pos^B) = \emptyset$.  Thus, $(S_r,\Gamma_r)$ is a good partition of $\Ch_\P(\pos)$.
\QED
%
%
%

\subsection*{Proof of Lemma \ref{time:maintainGraph}}

 We assume that we use the \textit{HDT} data structure \cite{HolmLichtenbergThorup:2001} to maintain the connected components of $\DG$, as nodes and edges are removed from it. Initializing HDT for \DG\ takes $\mathcal{O}(nk \log (nk))$ time; subsequent connectivity queries and edge and node deletions take $\Oh(\log^2 (nk))$ time \cite{HolmLichtenbergThorup:2001}.

Line \ref{algo:initGamma} of \ComputeSuccessor computes $\DG(\pos) \setminus S$ by successively deleting the edges from each label $\ell \in S$ to $\Ch_\P(\ell)$, and then delete $\ell$ itself.  Note that, some of these deletions may have already been performed for some ancestor position of \pos, where that label was also exposed. We say this type of exposed labels is \emph{old}. We refer to the labels that are exposed for the first time in $\pos$ as \emph{new} labels.  For each position \pos\ considered in line \ref{algo:initGamma} of \ComputeSuccessor, we only need to delete edges from each new label $\ell$ in \pos\ (and then delete $\ell$ itself).  Therefore, each vertex and edge of $\DG$ is deleted at most once, and the total number of vertex and edge deletions in $\Oh(nk)$ over the entire execution of \BuildAST, for a total of $\Oh(nk \log^2(nk))$ time.

Whenever an each edge deletion splits up a connected component, $\Ch_\P(\pos)$ is itself spilt, and we need to update the associated information.  We can do so  in $\mathcal{O}(M_p \log M_p)$ time by scanning the smaller of the two new connected components, as done in earlier papers \cite{FBLiu2019,FBGSV2015}.  We omit the details.

The \textbf{while} loop of lines \ref{algo:getSuccWhile}--\ref{algo:getSuccWhileEnd} of \ComputeSuccessor merges some of the components produced by line \ref{algo:initGamma}.  These operations do not modify the display graph.  We deal with these operations in Lemma \ref{time:getDec}.
\QED

\subsection*{Proof of Lemma \ref{time:initGD}}

To build sets $S$ and $K$ in line \ref{algo:initSK}, we scan each $\pos_i$ in $\pos$ for each $i \in [k]$. Given $\ell \in \pos_i$, update $\ell.\expose$ with $i$ and test if $\ell$ is exposed. Suppose $\pos$ has a parent position $\pos^*$.  Then, exposed label $\ell\ \in \pos_i$ is \textit{new} if $\pos_i \neq \pos^*_i$. This step takes $\mathcal{O}(k)$ time.

Now, consider Line \ref{algo:initGamma}.  
To find $\Gamma$, we scan each label $a \in \Ch_{T_i}(\pos_i)$ for each $i \in [k]$ and retrieve the set $A \in \Gamma$ that contains $a$ using the mapping from $\Ch_\P(\pos)$ to $\Gamma$. 
The entire process takes $\mathcal{O}(\sum_{i \in [k]} d_i)$ time. 
\QED

\subsection*{Proof of Lemma \ref{time:getDec}}

By Lemma \ref{lemma:computesuccessor} the \textbf{while} loop iterates $\mathcal{O}(k)$ times.  We complete the proof by showing that each iteration  takes $\mathcal{O}(\sum_{i \in [k]} d_i)$ time.  We rely on  the following below, which follows from the fact that, in line \ref{algo:initGamma} of \ComputeSuccessor, $\DG(\pos) \setminus S$ is obtained by deleting at most $\sum_{i \in [k]} d_i$ edges from $\DG(\pos)$.  

\begin{observation}\label{obs:sizeGamma}
$|\Gamma| \le \sum_{i \in [k]} d_i$.
\end{observation}

For each set $A \in \Gamma$, we maintain a count, initialized to $0$. By Observation \ref{obs:sizeGamma}, the total time to initialize the counts is $\mathcal{O}(\sum_{i \in K} d_i)$ per iteration.
To search for a bad label, for each $i \in K$, we scan each $a \in \Ch_{T_i}(\pos_i)$, and increase the count of the set $A$ to which $a$ belongs. If the count for any set $A \in \Gamma$ exceeds one, then $\ell \in \pos_i$ is a bad label and the search ends. 

Next, we consider the time taken by the body of the \textbf{while} loop.  Retrieving $K' = \ell.\expose$ in Line \ref{algo:retrieveKp} takes constant time. By Observation \ref{obs:sizeGamma} and the fact that  we have constant time access to mappings, building $\Gamma'$ in line \ref{algo:buildGammap}  takes $\mathcal{O}(\sum_{i \in K'} d_i)$ time.  

We compute the union of the sets in $\Gamma'$ in line \ref{algo:unions} as follows. We initialize $B$ to the empty set, and then successively consider each $A \in \Gamma'$. At each step, we append every child label $a$ from a non-empty entry in the representation of $A$ to the corresponding entry in $B$, and change the mapping of $a$  to $B$.  Given our representation of the sets in $\Gamma$, this process takes $\mathcal{O}(\sum_{i \in [k]} d_i)$ time in each iteration of the \textbf{while} loop.


Updating $\Gamma$ in Line \ref{algo:updateGamma} requires removing every $A \in \Gamma'$ from $\Gamma$ and then adding $B$. The time spent on updates is $\mathcal{O}(|\Gamma'|)$, which is $\mathcal{O}(\sum_{i \in K'} d_i)$.  Finally, updating $S$ in Line \ref{algo:getSuccWhileEnd} takes constant time and updating $K$ takes $\mathcal{O}(|K'|)$ time. 
\QED

\subsection*{Proof of Theorem \ref{thm:totalTime}}

By Lemmas \ref{time:initGD} and \ref{time:getDec}, lines \ref{algo:initSK}--\ref{algo:getSuccWhileEnd} of \ComputeSuccessor, take $\mathcal{O}(k\sum_{i \in [k]} d_i)$ time. By Theorem  \ref{theo:buildast}, \ComputeSuccessor is invoked $\mathcal{O}(n)$ times.  Thus, lines \ref{algo:initSK}--\ref{algo:getSuccWhileEnd} of \ComputeSuccessor take $\mathcal{O}(nk\sum_{i \in [k]} d_i)$ time the entire execution of \BuildAST.

For each $A \in \Gamma$, the \textbf{foreach} loop of lines \ref{algo:foreach}--\ref{algo:updatePi}  obtains the corresponding successor position using Equation \eqref{eqn:posDefn} in $\Oh(k)$ time. Since \BuildAST generates $\Oh(n)$ positions, the total time spent on the \textbf{foreach} loop of lines \ref{algo:foreach}--\ref{algo:updatePi} of \ComputeSuccessor over the entire execution of \BuildAST is $\Oh(nk)$.

To summarize,  \BuildAST takes $\mathcal{O}(nk\sum_{i \in [k]} d_i)$ to compute successors and, by Lemma \ref{time:maintainGraph}, $\mathcal{O}(nk \log^2 (nk))$ time to maintain the display graph.  The claimed time bound follows.
\QED

\end{document}